\documentclass[a4paper,10pt]{article}
\title{Isomorphism testing of rooted trees in linear time}
\author{Anna Lindeberg}
\date{}

\usepackage{amsmath,amsthm, amsfonts}
\makeatletter
\def\thmhead@plain#1#2#3{%
  \thmname{#1}\thmnumber{\@ifnotempty{#1}{ }\@upn{#2}}%
  \thmnote{ {\the\thm@notefont#3}}}
\let\thmhead\thmhead@plain
\makeatother

\newtheorem{theorem}{Theorem}[section] 
\newtheorem{lemma}[theorem]{Lemma} 

\theoremstyle{definition}

\newcommand{\alg}[1]{\textup{\texttt{#1}}}

\newcommand{\levelCDL}{\mathcal{L}}
\DeclareMathOperator{\lvl}{level}
\DeclareMathOperator{\parent}{parent}
\DeclareMathOperator{\child}{child}
\DeclareMathOperator{\cdl}{\ell}

\DeclareMathOperator{\lexleq}{\leq_{lex}}
\DeclareMathOperator{\height}{height}

\usepackage{graphicx}
\usepackage[dvipsnames]{xcolor}
\usepackage[ruled,vlined,linesnumbered]{algorithm2e}
\usepackage{tikz}
\usetikzlibrary{fit, calc}
\tikzset{vertex/.style={minimum size=3pt, inner sep = 0pt, fill,draw,circle}}
\usepackage[margin=44mm, top=4cm]{geometry}

\usepackage[backend=biber]{biblatex}
\usepackage{hyperref}

\usepackage{enumitem}
\usepackage{dirtytalk}

\begin{document}
\maketitle

\begin{abstract}
  The AHU-algorithm solves the computationally difficult graph isomorphism problem for rooted trees, and does so with a linear time complexity. Although the AHU-algorithm has remained state of the art for almost 50 years, it has been criticized for being unclearly presented, and no complete proof of correctness has been given. In this text, that gap is filled: we formalize the algorithm's main point of assigning and compressing labels to provide a characterization of isomorphic rooted trees, and then proceed with proving the correctness and optimal runtime of the AHU-algorithm.
\end{abstract}

\section{Introduction}
There is a rich collection of algebraic and combinatorial objects that consists of an underlying set, and some structure on that set. Often, the structural properties are what excites mathematicians, rather than the set. To capture the concept of two objects \say{having the same structure}, one often uses some kind of \emph{isomorphism}: a map that preserves structure, but not specific elements. Unfortunately, isomorphism are notoriously tricky to establish, in the sense that often the only option is to provide an explicit bijective map and prove that it indeed is an isomorphism.

The situation is not much better in an algorithmic setting. The well-studied decision problem of determining if two undirected graphs are isomorphic has neither been proved to be NP-complete nor to be polynomial time-solvable \cite{Schöning1988}. The most efficient algorithms for the graph isomorphism problem are quasi-polynomial in a worst-case scenario, but can be implemented to work rather well in practice, see \cite{McKay2014} for a good overview of the numerous algorithms that exist.

Despite these difficulties, in 1974 an algorithm for testing wether or not two rooted trees are isomorphic was given by Aho, Hopcroft, and Ullman \cite{Aho1974}. The algorithm, since then named the AHU-algorithm, runs in time proportional to the number of vertices of the input and has remained state of the art for almost fifty years. Still, it has been called \say{Utterly opaque. Even on second or third reading.} in \cite[p.255]{Campbell1991}. In other words: even if it is optimal in terms of efficiency, much of its details are left for the reader to figure out on their own.

Although no other linear time algorithm for isomorphism checking of rooted trees exist, people have tried to simplify and clarify the AHU-algorithm over the years. In \cite{Campbell1991}, Campbell and Radford does an excellent job at guiding the reader through the steps that lead to the full AHU-algorithm, developing intuition about why certain simplifications cannot be made. In \cite[Sec.~4.1.2]{Valiente2002}, a simpler and intuitively sound algorithm with a $O(n^2)$ time complexity is given. Moreover, the first new contribution to the AHU-algorithm since its introduction can be found in the recent preprint of Ingels \cite{Ingels2023}, where the somewhat tricky step of carefully sorting lists of tuples is replaced with multiplication of prime number (at the additional cost of generating primes). Still, this adapted AHU-algorithm has the same time complexity as the original.

The purpose of this text is twofold: first, we will explicitly prove the correctness and runtime results hinted at both in \cite{Aho1974} and \cite{Campbell1991} and thus, so to speak, fill the gap between the consensus of the correctness of the AHU-algorithm and what has been put in print. Secondly, the formalization of the mathematics behind the algorithm will, hopefully, complement \cite{Campbell1991} with a phrasing of the AHU-algorithm that lies rather close to how it may be implemented. In the experience of the author, formalization can, at times, also lead to intuition.

This paper is structured as follows: after a brief background on graphs, trees and lexicographical orders in Section~2, we prove a characterization of isomorphic rooted trees in Section~3.1, and establish the correctness and runtime of the AHU-algorithm in Section~3.2. A few closing remarks are given in Section~4.

\section{Preliminaries}
In this text, familiarity with both graph theory and rudimentary complexity theory is assumed, but we nonetheless introduce the graph theoretic terminology that will be used to make sure that no notation-related confusion arises.

\paragraph*{Graphs} We consider finite and undirected graphs $G=(V,E)$, where $V(G):=V$ is a nonempty set of \emph{vertices} of $G$, and $E(G):=E$ are the \emph{edges} of $G$. If $\{v,u\}\in E$, then the vertices $v$ and $u$ are said to be \emph{adjacent}. The \emph{degree} of a vertex $v$ is the number of vertices adjacent to $v$.
A \emph{path} in a graph $G=(V,E)$ is a non-repeating and ordered sequence $v_1v_2\ldots v_k$ of vertices in $V$ such that $\{v_i,v_{i+1}\}\in E$ for $i=1,2,\ldots,k-1$. In particular, the path $v_1v_2\ldots v_k$ is a \emph{path from $v_1$ to $v_k$} or, for short, a \emph{$v_1v_k$-path}. The \emph{length} of the path $v_1v_2\ldots v_k$ is the number of edges in it, that is, the number $k-1$. The \emph{distance} between two vertices $v$ and $u$ of a graph $G$, is the shortest possible length of a $uv$-path in $G$.

\paragraph*{Rooted trees}   A \emph{tree} is a graph $T=(V,E)$ in which there exist precisely one $uv$-path for each pair of distinct vertices $u,v\in V$. A \emph{rooted tree} is, formally speaking, an ordered pair $(T,r)$, where $T$ is a tree and $r\in V(T)$ is a distinguished vertex called the \emph{root} (of $T$). We will at times just state that $T$ is a rooted tree, leaving the presence of the root implicit. 

The cardinality of $V$ in a (rooted) tree $T=(V,E)$ is called the \emph{order} of $T$. It is a well-known fact that if $T$ is a tree with $n$ vertices, then it has precisely $n-1$ edges. In an algorithmic setting, this means that the number of edges is always proportional to the number of vertices, and any linear time complexity $O(|V|+|E|)$ reduces to $O(|V|)=O(n)$.

Let $T=(V,E)$ be a rooted tree with root $r$. If $v\in V$ has degree one or zero, then $v$ is said to be a \emph{leaf}. If $v$ is not a leaf, then it is called an \emph{inner vertex}. In particular, the so-called \emph{trivial tree} $T=(V,\emptyset)$ where $|V|=1$ has only one leaf and no inner vertices. Note that by definition of trees, there is a unique path $v_0v_1\ldots v_k$ for $v_0=r$ and any vertex $v=v_k$. In that case, $v_k$ is said to be a \emph{child} of $v_{k-1}$, whereas $v_{k-1}$ is the \emph{parent} of $v_k$. By the uniqueness of the $rv$-path in $T$, it is clear that a vertex $v$ may have any number of children, but the parent of $v\neq r$ is always a single vertex denoted $\parent(v)$. The (possibly empty) set of children of a vertex $v$ is denoted with $\child(v)$. Moreover, a vertex $v$ is an \emph{descendant} of a vertex $u$ if $u$ lies on the unique $rv$-path in $T$.

The \emph{height} of a rooted tree $(T,r)$, denoted $\height(T)$, is the maximum length of a $rv$-path in $T$, where $v$ is any vertex of $T$. If $d$ denote the distance from the root $r$ to some vertex $v$ of $T$, then we define the \emph{level} of $v$ as $\lvl(v):=\height(T)-d$. We also say that \emph{$v$ is at level $\lvl(v)$}. In particular, the root is the only vertex whose level is $\height(T)$, and every vertex at level 0 will be a leaf. Note, however, that there may be leafs in other levels as well. Lastly, for some rooted tree $T$, we put $T_i\subset V(T)$ to be the set $T_i=\{v\in V(T)\,:\,\lvl(v)=i\}.$

At times, when context is not sufficiently precise, we specify the underlying rooted tree $T$ with an index so that $\parent_T(v)$, $\child_T(v)$, and $\lvl_T(v)$ denotes the same as $\parent(v)$, $\child(v)$, respectively $\lvl(v)$.

Two rooted trees $(T,r)$ and $(T',r')$ are \emph{isomorphic} if there exists a bijective map $\varphi:V(T)\to V(T')$ such that 
\begin{enumerate}[label = (\roman*)]
  \item for all $u,v\in V(T)$ we have $\{u,v\}\in E(T)$ if and only if $\{\varphi(u),\varphi(v)\}\in E(T')$, and such that 
  \item $\varphi(r)=r'$. 
\end{enumerate}

\paragraph*{Sets and the lexicographical order}
We let $[n]$ denote the set $\{0,1,\ldots,n\}$. A \emph{partition} of a set $X$ is a family of disjoint subsets $\{X_i\}_{i=1}^k$ of $X$ such that $\cup_{i=1}^k X_i=X$.

When $S$ is a set, $S^*$ denotes the set of finite tuples with components in $S$. Any linear order $\leq$ of $S$ can then be extended to the \emph{lexicographical order} $\lexleq$ on $S^*$ by letting $(s_1,s_2,\ldots,s_k)\lexleq (t_1,t_2,\ldots,t_l)$ if and only if (a) $k\leq l$ and $s_i=t_i$ for all $1\leq i\leq k$ or (b) there is some index $j$ such that $s_i=t_i$ for all $1\leq i< j$ and $s_j<t_j$.  In particular, $S^*$ contains the empty tuple $()$, which is smaller than any other tuple of $S^*$. When $S=[n]$ with its canonical ordering, we for example have that
\[()\lexleq(0,1)\lexleq(0,1,1,1)\lexleq(0,1,2,0)\lexleq(3,1).\]

The authors of the AHU-algorithm show how a list containing $n$ tuples in $[k]^*$ can be rearranged in a lexicographical order efficiently, using an algorithm we call \alg{LEXsort}, originally given in \cite[Alg.~3.2]{Aho1974}. This algorithm is a variation of the more well-known radix-sort (see e.g. \cite[Sec.~8.3]{Cormen+2009} for a modern treatment), where the tuples in the list may have different numbers of components. 

\begin{theorem}[{\cite[Thm.~3.2]{Aho1974}}]\label{thm:LEXsort}
  The algorithm \alg{LEXsort} can be implemented in $O(n+l_{\text{tot}})$ time for an input list with $n$ tuples in $[k]^*$, where $k=O(n)$ and $l_{\text{tot}}$ is the sum over the lengths of all tuples.
\end{theorem}

\section{The AHU-algorithm}
The AHU-algorithm is a bottom-up procedure, working level by level in the two input trees simultaneously. In essence, each vertex in the current level is first assigned a tuple label that encodes a condensed history of its descendants. The label of a vertex is recursively constructed from the labels of its children. In particular, the labels of two vertices $v$ and $u$ in the same level are equal if and only if the tree containing the descendants of $v$ is isomorphic to the tree containing the descendants of $u$. These labels are then put into one sorted list per tree, and the lists are compared. If not equal, then the trees cannot be isomorphic and the algorithm terminates.
Otherwise, the labels are shortened, and the algorithm may proceed with the next iteration. The trees are concluded to be isomorphic if and only if the roots obtain the same label. Note that the step where the labels are shortened is not strictly necessary for the correctness of the algorithm but it is, however, crucial to obtain a linear run time. Before we see the full AHU-algorithm, we now provide a characterization of isomorphic rooted trees.

\subsection{Characterization of isomorphic rooted trees}

Let $S=\{s_1,s_2,\ldots,s_k\}$ be a finite subset of $[n]^*$ of cardinality $k$ such that $s_1\lexleq \ldots\lexleq s_k$. The \emph{label-compression} of $S$ is the map $c: S\to [k]$ defined by $c(s_1):=0$ and $c(s_i):=c(s_{i-1})+1$ for $i=2,3,\ldots,k$. In other words, $c$ assigns consecutive integers to consecutive (w.r.t. $\lexleq$) tuples of $S$, starting at 0.

Now, the \emph{compressed descendant labeling} $\cdl:V(T)\to[n]^*$ of a rooted tree $T$ of order $n$ is defined recursively in terms of the levels of $T$. If $v$ is a leaf of $T$ (in any level), then $\cdl(v):=()$. As level zero of a tree contains only leafs, we may assume $\cdl(v)$ is defined for every vertex $v$ in level $i$, before defining $\cdl(u)$ of a vertex $u$ in level $i+1$ of $T$. Thus, assume $u\in T_{i+1}$ is a inner vertex with children $u_1$, $u_2$, ..., $u_k$ such that \[\cdl(u_1)\lexleq\cdl(u_2)\lexleq\ldots\lexleq\cdl(u_k),\] and let $c_i$ be the label-compression of the set $S_i=\{\cdl(v)\,:\, v\in T_i\}$.
With that, define $\cdl(u):=(c_i(\cdl(u_1)),c_i(\cdl(u_2)),\ldots,c_i(\cdl(u_k))).$ Note that $c_i$ maps tuples to integers in $[m]$, where $m=|T_i|$, and that the $\cdl(u)$ for $u\in T_i$ is a tuple whose components all lie in the image of $c_i$. Thus $\cdl(u)$ will itself be a tuple in $[m]^*$, and the label-compression of $S_{i+1}$ is well-defined. 

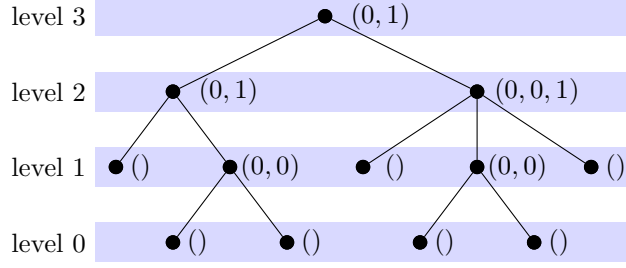
\begin{figure}
  \begin{center}
  \begin{tikzpicture}[level distance=10mm, level 1/.style={sibling distance=40mm},
    level 2/.style={sibling distance=15mm},
    vertex/.style={minimum size=5pt, inner sep = 0pt, fill,draw,circle},
    leaf/.style={label={[name=#1]below:$#1$}},
    every fit/.style={rectangle,fill=blue!15,inner sep=4pt}]
\node (L3left) at (-27.5mm,0) {};
\node (L3right) at (38mm,0) {};
\node (L2left) at ($(L3left)-(0,10mm)$) {};
\node (L2right) at ($(L3right)-(0,10mm)$) {};
\node (L1left) at ($(L3left)-(0,20mm)$) {};
\node (L1right) at ($(L3right)-(0,20mm)$) {};
\node (L0left) at ($(L3left)-(0,30mm)$) {};
\node (L0right) at ($(L3right)-(0,30mm)$) {};
\node[fit= (L0left) (L0right), label= left: level 0] {};
\node[fit= (L1left) (L1right), label= left: level 1] {};
\node[fit= (L2left) (L2right), label= left: level 2] {};
\node[fit= (L3left) (L3right), label= left: level 3] {};
%
\tikzset{nodes = vertex}
\node[vertex] (root) {}
child {node[vertex] (l21) {}
  child {node (l11) {}}
  child {node[vertex] (l12) {}
    child {node (l01) {}}
    child {node (l02) {}}}}
child {node[vertex] (r21) {}
  child {node (r11) {}}
  child {node[vertex] (r12) {}
    child {node (r01) {}}
    child {node (r02) {}}}
  child {node (r13) {}}
};
\node[label = {right:$()$}] at (l01) {};
\node[label = {right:$()$}] at (l02) {};
\node[label = {right:$()$}] at (r01) {};
\node[label = {right:$()$}] at (r02) {};
\node[label = {right:$()$}] at (l11) {};
\node[label = {right:$(0,0)$}] at (l12) {};
\node[label = {[label distance=1mm]right:$()$}] at (r11) {};
\node[label = {right:$(0,0)$}] at (r12) {};
\node[label = {right:$()$}] at (r13) {};
\node[label = {[label distance=2mm]right:$(0,1)$}] at (l21) {};
\node[label = {[label distance=1mm]right:$(0,0,1)$}] at (r21) {};
\node[label = {[label distance=2mm]right:$(0,1)$}] at (root) {};
\end{tikzpicture}
\end{center}
\caption{A rooted tree $T$, where the compressed descendant label $\cdl(v)$ is given alongside each vertex.}
\end{figure}

An example is in place. In the rooted tree $T$ given in Figure~1, each leaf $v$ satisfy $\cdl(v)=()$. The label-compression $c_0$ of the set of labels assigned to vertices in $T_0$ just maps $()\mapsto0$, so that the two inner vertices of $T_1$ gets the label $(c(()),c(()))=(0,0)$. There are two vertices $u$ and $v$ in level 2 of $T$, where $\cdl(u)=(0,1)$ and $\cdl(v)=(0,0,1)$. Since $(0,0,1)\lexleq(0,1)$, the label compression $c_2$ of $\{\cdl(v),\cdl(u)\}$ maps $c_2((0,0,1))=0$ and $c_2((0,1))=1$. Hence, $\cdl(r)=(0,1)$. Note that the equality $\cdl(u)=(0,1)=\cdl(r)$ is, more or less, a coincidence. The label $(0,1)$ should only be interpreted as: this vertex has two children, and these two children have different descendant histories.

Intuitively speaking, two rooted trees are isomorphic if and only if they are isomorphic "level by level". This is the essential idea of the AHU-algorithm. The formal statement depends on the compressed descendant labeling of the trees in question, but we need one last piece of additional notation: for a rooted tree $T=(V,E)$ with compressed descendant labeling $\cdl$, define $\levelCDL(W)$ as the \emph{multiset} of compressed descendant labels $\cdl(v)$ of all vertices $v\in W$, where $W$ is any subset of $V$. Moreover, we say that a pair of rooted trees $T$ and $T'$ \emph{satisfies the level-label condition} if $\levelCDL(T_i)=\levelCDL(T'_i)$ for $i=0,1,\ldots, h$, where $h=\max\{\height(T),\height(T')\}$. Note that the level-label condition implies that $\height(T)=\height(T')$. Moreover, the label-compressions defined above are formally defined on \emph{sets} of labels of vertices in $T_i$, but for simplicity we sometimes state that $c$ is the label-compression of $\levelCDL(T_i)$ for some level $i$; formally speaking, $c$ is then the label-compression of the underlying set of the multiset $\levelCDL(T_i)$. We begin with a lemma of technical nature.

\begin{lemma}\label{lem:level-label implication}
  If $(T,r)$ and $(T',r')$ are rooted trees that satisfy the level-label condition with common height $h$, then there exist a collection of bijective maps $\varphi_i:T_i\to T'_i$ for $i\in[h]$ such that $\varphi_{i+1}(\parent_T(v))=\parent_{T'}(\varphi_i(v))$ for all $v\in T_i$ and all $i\in[h-1]$.
\end{lemma}
\begin{proof}
  Let $(T,r)$ and $(T',r')$ be rooted trees of common height $h$ that satisfy the level-label condition. Throughout this proof, let $\cdl_T$ and $\cdl_{T'}$ denote the compressed descendant labeling of the respective trees. Moreover, let $c_i$ denote the label-compression of the multiset $\levelCDL(T_i)$ and $c'_i$ denote the label-compression of the multiset $\levelCDL(T'_i)$, for $i\in[h]$. 
  
  We will now recursively define maps $\varphi_i:T_i\to T'_i$ for $i=h, h-1,\ldots, 0.$
  Since $T_{h}=\{r\}$ and $T'_{h}=\{r'\}$ we put $\varphi_h(r):=r'$. As $\levelCDL(T_{h})=\levelCDL(T'_{h})$ it is clear that $\cdl_T(r)=\cdl_{T'}(r')=\cdl_{T'}(\varphi_h(r))$, i.e. $\cdl_T(v)=\cdl_{T'}(\varphi_{i}(v))$ for all $v\in T_h$. 

  When defining $\varphi_i$ for $i\in[h-1]$, it is thus fair to assume that $\varphi_{i+1}$ is defined such that $\cdl_T(v)=\cdl_{T'}(\varphi_{i+1}(v))$ for all $v\in T_{i+1}$. Note that since $\levelCDL(T_i)=\levelCDL(T'_i)$, the label-compressions $c_i$ and $c'_i$ satisfy
  \begin{equation*}\label{eq: ci=ci'}
    c_i(x)=c'_i(y)\iff x=y \quad \text{for all } x,y\in\levelCDL(T_i)=\levelCDL(T'_i),
  \end{equation*}
  which means that the condition of $\cdl_T(v)=\cdl_{T'}(\varphi_{i+1}(v))$ can be equivalently written as 
  \begin{equation}\label{eq:level label formulation}
  \levelCDL(\child_T(v))=\levelCDL(\child_{T'}(\varphi_{i+1}(v))),
  \end{equation}
  as the respective labels are equal component-wise. In extension it holds that for each $u\in \child_T(v)$ where $v\in T_{i+1}$, there exist some $u'\in\child_{T'}(\varphi_{i+1}(v))$ such that $\cdl_T(u)=\cdl_{T'}(u')$. 
  
  Moreover, since $\{\child_T(v)\}_{v\in T_{i+1}}$ is a partition of $T_i$, we may define $\varphi_i(u)$ for $u\in T_i$ by the following two properties; $\varphi_i(u):=u'$, where $\cdl_T(u)=\cdl_{T'}(u')$ and $\parent_{T'}(u')=\varphi_{i+1}(\parent_T(u))$. If there are multiple choices of $u'$ (i.e. the multiplicity of $\cdl_{T}(u)$ in $\levelCDL(\child_{T'}(\parent_T(u)))$ is $\geq2$), then $\varphi_i$ may be picked so that it is a bijection of $T_i$ and $T'_i$. This is indeed possible, as $\levelCDL(T_i)$ and $\levelCDL(T'_i)$ are equal as multisets. To be precise, if, for some fixed $v\in T_{i+1}$, $\{u_1,\ldots,u_r\}\subseteq \child_T(v)$ and $\{u'_1,\ldots,u'_r\}\subseteq \child_{T'}(\varphi_{i+1}(v))$ are maximal subsets (in terms of cardinality) such that \[\cdl_T(u_1)=\ldots=\cdl_T(u_r)=\cdl_{T'}(u'_1)=\ldots=\cdl_{T'}(u'_r)\]
  is satisfied, then we define $\varphi_i(u_j):=u'_j$ for $j=1,2,\ldots, r$. We have thus shown that if there is a bijective map $\varphi_{i+1}:T_{i+1}\to T'_{i+1}$ that satisfies $\cdl_T(v)=\cdl_{T'}(\varphi_{i+1}(v))$ for all $v\in T_{i+1}$, then we can define a bijective map $\varphi_i:T_i\to T'_i$ that satisfies $\cdl_T(v)=\cdl_{T'}(\varphi_{i}(v))$ for all $v\in T_i$ and, moreover, $\varphi_{i+1}(\parent_T(v))=\parent_{T'}(\varphi_i(v))$. In conclusion, a family of bijective maps $\{\varphi_i\}_{i\in[h]}$ with the required property always exist.
\end{proof}

We now characterize isomorphic rooted trees as those that satisfy the level-label condition. This theorem highly resembles Observation~8 in \cite{Campbell1991}, although that statement was given without proof.

\begin{theorem}\label{thm:characterization}
  Two rooted trees are isomorphic if and only if they satisfy the level-label condition.
\end{theorem}
\begin{proof}

  First assume that $\varphi:V(T)\to V(T')$ is an isomorphism of the rooted trees $T$ and $T'$. It suffices to show that the respective compressed descendant labeling $\cdl_T$ and $\cdl_{T'}$ satisfies $\cdl_T(v)=\cdl_{T'}(\varphi(v))$ for all $v\in V(T)$, since $\varphi$ is a bijection such that $\lvl_T(v)=\lvl_{T'}(\varphi(v))$. 
  The proof proceeds by induction on the level of a vertex $v$ in $T$. For the base case, assume $v\in T_0$. Then $v$, as well as $\varphi(v)$, is a leaf and $\cdl_T(v)=()=\cdl_{T'}(\varphi(v))$ follows immediately.
  
  Thus assume there is some level $i\geq0$ such that for all $u\in T_i$ we have that $\cdl_T(u)=\cdl_{T'}(\varphi(u))$. In particular, $\levelCDL(T_i)=\levelCDL(T_i')$. Let $v$ be a vertex in $T_{i+1}$. If $v$ is a leaf, then once again we have that $\cdl_T(v)=()=\cdl_{T'}(\varphi(v))$, as $\varphi(v)$ is a leaf of $T'$. If $v$ is not a leaf, then $v$ has $k\geq1$ children $v_1$, $v_2$, ..., and $v_k$. As $\varphi$ is an isomorphism, the children of $\varphi(v)$ are precisely $\varphi(v_1)$, $\varphi(v_2)$, ..., and $\varphi(v_k)$. By the induction hypothesis, we have
  \begin{equation}\label{eq:label equality}
    \cdl_T(v_j)=\cdl_{T'}(\varphi(v_j))\quad\text{for }j=1,2,\ldots,k.
  \end{equation}
  Moreover, since $\levelCDL(T_i)=\levelCDL(T'_i)$ the respective label-compressions $c_i$ and $c_i'$ (of $\levelCDL(T_i)$ and $\levelCDL(T'_i)$, respectively) are equal, i.e. $c_i(x)=c_i'(x)$ for all $x$. Combining this fact with \eqref{eq:label equality} and the definition of $\cdl_T$ and $\cdl_{T'}$ we obtain
  \begin{align*}
    \cdl_T(v)&=(c_i(\cdl_T(v_1)),c_i(\cdl_T(v_2)),\ldots,c_i(\cdl_T(v_k)))\\
    &=(c_i'(\cdl_{T'}(\varphi(v_1))),c_i'(\cdl_{T'}(\varphi(v_2))),\ldots,c_i'(\cdl_{T'}(\varphi(v_k))))\\
    &=\cdl_{T'}(\varphi(v)).
  \end{align*}
  Hence the principle of induction implies that $\cdl_T(v)=\cdl_{T'}(\varphi(v))$ for all $v\in V(T)$. In conclusion, $\levelCDL(T_i)=\levelCDL(T'_i)$ for all $i\in[\max\{\height(T),\height(T')\}]$.

  For the converse statement, assume $(T,r)$ and $(T',r')$ are rooted trees that satisfy the level-label condition, with common height $h$. By Lemma~\ref{lem:level-label implication} there is a collection of bijective maps $\varphi_i:T_i\to T'_i$ where $i\in [h]$, such that \[\varphi_{i+1}(\parent_T(v))=\parent_{T'}(\varphi_i(v))\text{ for all }v\in T_i\text{ and }i\in[h-1].\] Since $\{T_i\}_{i=0}^h$ is a partition of $V(T)$ (and $\{T'_i\}_{i=0}^h$ is a partition of $V(T')$) we may define a bijective map $\varphi:V(T)\to V(T')$ by putting $\varphi(v):=\varphi_{\lvl_T(v)}(v)$ for all $v\in V(T)$. 
  
  We now show that $\varphi$ is an isomorphism. It is clear that $\varphi(r)=\varphi_h(r)=r'$, as $\varphi_h:\{r\}\to\{r'\}$. Now, let $\{v,u\}$ be an edge of $T$. Clearly, $v$ and $u$ can be picked so that $u=\parent_T(v)$. Then, with $i:=\lvl_T(v)$, it follows that
  \[\{\varphi(v),\varphi(\parent_T(v))\}=\{\varphi_i(v),\varphi_{i+1}(\parent_T(v))\}=\{\varphi_i(v),\parent_{T'}(\varphi_{i}(v))\}\]
  by assumption on $\varphi_i$ and $\varphi_{i+1}$. Clearly, $\{\varphi_i(v),\parent_{T'}(\varphi_{i}(v))\}$ is an edge of $T'$. If, instead, $\{x,y\}$ is an edge of $T'$, then we may without loss of generality assume that $x\in T'_i$ for $i=\lvl_{T'}(x)$, $x=\varphi_i(v)$ for some $v\in T_i$ and $y=\parent_{T'}(x)$. Then the property of $\varphi_i$ and $\varphi_{i+1}$ implies
  \begin{align*}
    \{\varphi^{-1}(x),\varphi^{-1}(\parent_{T'}(x))\}&=\{\varphi_i^{-1}(\varphi_i(v)),\varphi_{i+1}^{-1}(\parent_{T'}(\varphi_i(v)))\}\\
  &=\{v,\varphi_{i+1}^{-1}(\varphi_{i+1}(\parent_{T}(v)))\}\\
  &=\{v,\parent_T(v)\},
  \end{align*}
  where $\{v,\parent_T(v)\}\in E(T)$ is obvious. In conclusion, we have $\{u,v\}\in E(T)$ if and only if $\{\varphi(u),\varphi(v)\}\in E(T')$, so that $\varphi$ is an isomorphism of $T$ and $T'$.
\end{proof}

\subsection{Correctness and runtime}
We present the full AHU-algorithm in Algorithm~\ref{alg:AHU}. On line~1, the trees vertices are partitioned per level; this can be done with a simple breadth-first search as noted in e.g. \cite[pp.39-41]{Golumbic1980}. In particular, the heights of the two input trees can be computed alongside the levels of respective vertex sets, so that the check on lines~2--5 becomes trivial. The lists $L$ and $L'$ used in the remainder of the AHU-algorithm carry the labels of vertices in a current level, and in a practical setting they need to contain vertex-label pairs rather than labels only. However, the underlying vertices may be considered to be a sort of \say{satellite data} which is of less concern in the mathematical analysis of the algorithm. For a more in-depth discussion about satellite data in sorting algorithms (which is equally applicable in this context), see \cite[p.147]{Cormen+2009}.

As earlier mentioned, the main part of the AHU-algorithm iterates over the levels, starting at level 0 and working upwards in both trees simultaneously. In essence, the \textbf{for}-loop of Algorithm~\ref{alg:AHU} checks that the level-label condition is satisfied for the two input-trees. The tuples $\cdl(v)$ are represented by lists. However, since multisets are difficult to compare in an algorithmic setting, Algorithm~\ref{alg:AHU} represents the multisets $\levelCDL(T_i)$ and $\levelCDL(T'_i)$ by lexicographically sorted lists instead. To be precise, these two multisets are represented by $L$ and $L'$, and they are correctly (as we will see) computed after executing lines 8--11 of Algorithm~\ref{alg:AHU}. However, in contrast to the strict mathematical definition of the compressed descendant labeling, the AHU-algorithm calculates $\cdl(v)$ assuming that the "condensed" labels of its children are known. More precisely, if $v$ is a vertex with children $v_1$, ..., $v_k$, then $\cdl(v)$ is computed from knowing the \emph{values} $c(\cdl(v_1))$, ..., $c(\cdl(v_k))$, rather than the values $\cdl(v_1)$, ..., $\cdl(v_k)$ and the label-compression $c$ of $\levelCDL(T_{\lvl(v)-1})$.  This distinction is made in the two subroutines \alg{CreateLabels} and \alg{CondenseLabels}, which we now study a bit more carefully. 

\begin{algorithm}
  \SetKwData{True}{true}\SetKwData{False}{false}\SetKwData{Level}{lvl}
  \SetKwFunction{treeheight}{height}\SetKwFunction{LEXsort}{LEXsort}
  \SetKwFunction{CondLab}{CondenseLabels}\SetKwFunction{CreateLab}{CreateLabels}
  \SetKwFunction{AddLeafs}{AddLeafLabels}
  \KwIn{Two rooted trees $T$ and $T'$}
  \KwOut{\True if $T$ and $T'$ are isomorphic, otherwise \False}
  \BlankLine
  Pre-process the trees so that $T_i$ and $T'_i$ are available\;
  \If{\treeheight{$T$} $\neq$ \treeheight{$T'$}}{
    \KwRet{\False}
  }\Else{
    $h\leftarrow$ \treeheight{$T$}\;
  }
  Initialize $L$ and $L'$ as empty lists\;
  \For{\Level $\leftarrow 0$ \KwTo $h$}{
    $L\leftarrow$ \CreateLab{$L$, $T_{\Level}$}\;
    $L'\leftarrow$ \CreateLab{$L'$, $T'_{\Level}$}\;
    $L\leftarrow$ \LEXsort{$L$}\;
    $L'\leftarrow$ \LEXsort{$L'$}\;
    \If{$L\neq L'$}{
      \KwRet{\False}
    }
    $L, L'\leftarrow$ \CondLab{$L$}\;
  }
  \KwRet{\True}

  \caption{The AHU-algorithm \alg{AHU}.}\label{alg:AHU}
\end{algorithm}

\begin{lemma}\label{lem:CreateLabels}
  Let $T$ be a rooted tree with compressed descendant labeling $\cdl$. Moreover, let $c_i$ denote the  label-compression of $\levelCDL(T_i)$ for some $i\in [\height(T)]$. If $L$ is a sorted list of integers containing the value $c_i(\cdl(v))$ for each $v\in T_{i}$, then $\alg{CreateLabels}(L,T_{i+1})$ outputs a list containing $\cdl(v)$ for each $v\in T_{i+1}$, in $O(|T_i|+|T_{i+1}|)$ time.
\end{lemma}

\begin{proof}
  If $v\in T_{i+1}$ is a leaf, then $\cdl(v)=()$, and it is clear that $L'$ contains the vertex-label pair $(v,[])$ after executing rows~2--4 of \alg{CreateLabels} (found in Algorithm~\ref{alg:create})\footnote{Recall that the tuples are represented with lists.}.  If $u\in T_{i+1}$ is not a leaf, and thus have children $u_1$, ..., $u_k$, then, by assumption, the integer values $c_i(\cdl(u_1))$, ..., $c_i(\cdl(u_k))$ appear in non-descending order in $L$. Hence the label $\cdl(u)$ of $u$ in $L'$ will be precisely $[c_i(\cdl(u_1)),\ldots,c_i(\cdl(u_k))]$ for $\cdl(u_1)\lexleq\ldots\lexleq\cdl(u_k)$, since $c_i$ is strictly increasing and elements are added to $\cdl(u)$ in order while running rows~5--10.
  Thus $L'$ contains $\cdl(v)$ for each $v\in T_{i+1}$. For the runtime, simply note that rows~2--4 of Algorithm~\ref{alg:create} take $O(|T_{i+1}|)$ time, and rows~5--10 take $O(|L|)$ time. Since $L$ contains as many integer labels as there are vertices in $T_i$, the combined time complexity is $O(|T_i|+|T_{i+1}|)$.
\end{proof}

To reiterate, \alg{CreateLabels} take a sorted list of the integers $c(\cdl(v))$ for $v\in T_i$ and outputs an unsorted list of the tuples $\cdl(u)$ for $u\in T_{i+1}$. As we will now see, \alg{CondenseLabels} takes a sorted list of tuples $\cdl(u)$ for $u\in T_{i+1}$ and outputs a sorted list of the integers $c(\cdl(u))$, where $c$ is the label-compression of $\levelCDL(T_{i+1})$. 

\begin{algorithm}
  \SetKwData{parent}{parent}
  \KwIn{List $L$ of vertex-label pairs for all vertices in some level $i$. The labels are integers.}
  \KwIn{A set $T_{i+1}$ of all vertices in level $i+1$}
  \KwOut{List $L'$ of vertex-label pairs for all vertices in level $i+1$}
  \BlankLine
  $L'\leftarrow$ empty list\;
  \For{\textup{vertex} $v$ \textup{\textbf{in} $T_{i+1}$}}{
    \If{$v$ \textup{is a leaf}}{
    Append new vertex-label pair $(v,[])$ to $L'$\;}
  }
  \For{\textup{vertex-label pair} ($v$, $k$) \textup{\textbf{in}} $L$}{
    $\parent\leftarrow$ parent of $v$\;
    \If{\parent is a vertex \textup{\textbf{in}} $L'$}{
      Append $k$ to the label of \parent in $L'$\;
    }\Else{
      Append new vertex-label pair $(\parent,[k])$ to $L'$\;
    }
  }
  \KwRet{$L'$}
  \caption{The subroutine \alg{CreateLabels}.}\label{alg:create}
\end{algorithm}

\begin{lemma}\label{lem:CondenseLabels}
  Let $L$ be a lexicographically sorted list of elements $\cdl(v)$ for $v\in T_i$ of some underlying rooted tree $T$, where $\cdl$ is the compressed descendant labeling of $T$ and $i\in[\height(T)]$. Then $\alg{CondenseLabels}(L)$ outputs a sorted list of the integers $c(\cdl(v))$ for $v\in T_i$, where $c$ is the label-compression of $\levelCDL(T_i)$. Moreover, it does so in $O(|T_i|+ |T_{i-1}|)$ time.
\end{lemma}
\begin{proof}
  The method \alg{CondenseLabels}, given in Algorithm~\ref{alg:condense}, clearly replaces consecutive and distinct labels in the given list $L$ by consecutive integers, starting at $0$. By assumption, the label at position $i$ of $L$ will be the tuple $\cdl(v_i)$ for $v_i\in T_i$, where the vertices $v_1,\ldots,v_k$ of $T_i$ are indexed such that $\cdl(v_1)\lexleq\ldots\lexleq\cdl(v_k)$. Hence the output $L'$ will contain the elements $c(\cdl(v_1))$, ..., $c(\cdl(v_k))$ in that order. 
  
  For the time complexity, note that the \textbf{for}-loop of Algorithm~\ref{alg:condense} runs in $|L|$ steps. In each step, the worst-case scenario would involve comparing each integer component of the current label to the previous label. With the constant-time operations of Algorithm~\ref{alg:condense} taken into account, the algorithm runs in $O(|L|+\sum_{i=1}^{k}|L_{i}|)$, where $|L_{i}|$ is the length of the $i$:th element of $L$. As $|L|=|T_i|$, and the sum of the lengths of the tuples in $L$ is the same as the total number of children of vertices in $T_i$, this yields a $O(|T_i|+|T_{i-1}|)$ runtime (where, for notational purposes, $|T_{-1}|=0$, even though there is no level $-1$).
\end{proof}

\begin{algorithm}
  \SetKwData{prevInt}{prevInt}\SetKwData{prevLabel}{prevLabel}
  \KwIn{List $L$ of vertex-label pairs}
  \KwOut{List $L'$ of vertex-label pairs, where $L'$ is ordered as $L$, and the labels are condensed}
  \BlankLine
  $L'\leftarrow$ empty list\;
  \prevInt$\leftarrow0$\;
  \prevLabel$\leftarrow$ first label of $L$\;
  \For{\textup{vertex-label pair} ($v$, $\ell$) \textup{\textbf{in}} $L$ in order}{
    \If{$\ell\neq$ \prevLabel}{
      \prevInt$\leftarrow\prevInt+1$\;
      \prevLabel$\leftarrow\ell$\;
    }
    Append vertex-label pair ($v$, \prevInt) to $L'$;
  }
  \KwRet{$L'$}
  \caption{The subroutine \alg{CondenseLabels}.}\label{alg:condense}
\end{algorithm}

\begin{theorem}
  The AHU-algorithm is correct. That is, it outputs \alg{True} if and only if the two input trees are isomorphic. Moreover, the AHU-algorithm may be implemented in a linear runtime.
\end{theorem}

\begin{proof}
Due to Theorem~\ref{thm:characterization} it suffices to show that $\alg{AHU}(T,T')$ outputs \alg{False} if and only if the rooted trees $(T,r)$ and $(T',r')$ fail the level-label condition. Note that $\alg{AHU}(T,T')$ correctly outputs \alg{False} (c.f. lines~2--3 of Alg.~\ref{alg:AHU}) if $\height(T)\neq\height(T')$. We may thus assume $\height(T)=\height(T')=:h$. 

For each $i\in[h]$, let $L_{(i)}$ denote the list stored in the variable $L$ after line~10 has been executed in iteration $i+1$ of the \textbf{for}-loop of Algorithm~\ref{alg:AHU}, i.e. in the iteration where the variable \alg{lvl} takes the value $i$. We show that $L_{(i)}$ is an ordered representation of the multiset $\levelCDL(T_i)$, by induction on $i$. The base case of $i=0$ is easily verified. 

Now, assume there is some $j\in[h-1]$ such that $L_{(j)}$ contains every element of $\levelCDL(T_j)$ in lexicographical order, and consider $L_{(j+1)}$. By assumption, Lemma~\ref{lem:CondenseLabels} may be applied and we conclude that after $\alg{CondenseLabels}(L_{(j)})$ on line~12 is executed, $L$ stores a sorted list $L_{(j)'}$ of integers. In particular, the $k$:th element of $L_{(j)'}$ equals $c(\cdl(v))$, where $\cdl(v)$ is the $k$:th element of $L_{(j)}$, for $k=1,2,\ldots,|L_{(j)}|=|L_{(j)'}|$. Thus, the assumptions of Lemma~\ref{lem:CreateLabels} are satisfied, and $L$ stores a list $L_{(j+1)'}$ of the labels $\cdl(v)$ for each $v\in T_{j+1}$ after executing $\alg{CreateLabels}(L_{(j)'})$ on line~8. By Theorem~\ref{thm:LEXsort}, $L_{(j+1)}$ is a lexicographically sorted representation of $\levelCDL(T_{j+1})$, after executing line~10. The principle of induction implies that $L_{(i)}$ represents $\levelCDL(T_i)$ for all $i\in[h]$. 

By analogous arguments the list $L'_{(i)}$ stored in $L'$ after running line~11 of iteration $i+1$ represents $\levelCDL(T'_i)$ for all $i\in [h]$. Since line~13 is executed if and only if $L_{(i)}\neq L'_{(i)}$ for some $i\in [h]$, it is executed if and only if $\levelCDL(T_i)\neq\levelCDL(T'_i)$ for some $i\in[h]$. That is, $\alg{AHU}(T,T')$ outputs \alg{False} for two trees of equal height $h$ if and only if $T$ and $T'$ does not satisfy the level-label condition.

For the runtime, suppose $T$ and $T'$ has $n$ respectively $n'$ vertices. If the heights of $T$ and $T'$ are different, then Algorithm~\ref{alg:AHU} terminates after line~3. Up to that point, the runtime is dominated by the breadth-first search of line~1, which by e.g. \cite{Golumbic1980} take $O(n)$ time for $T$ and $O(n')$ time for $T'$. The total runtime is thus, in this case, linear.

If $\height(T)=\height(T')=:h$, then let $n_i:=|T_i|$ and $n'_i:=|T'_i|$ for each $i\in[h]$, so that $\sum_{i=0}^{h}n_i=n$ and $\sum_{i=0}^{h}n'_i=n'$. For notational purposes, put $n_{-1}:=n'_{-1}:=0$. By Lemma~\ref{lem:CreateLabels} and Lemma~\ref{lem:CondenseLabels} the calls on lines~8,9, and 14 take $O(n_i+n_{i-1})$ respectively $O(n'_i+n_{i-1})$ time during the $(i+1)$:th iteration of the \textbf{for}-loop. Note that $L_{(i)}$ contains $n_i$ labels, and that each label $\cdl(v)$ in $L_{(i)}$ for $v\in T_i$ contains as many components as the vertex $v$ has children. Moreover, the total number of children of all vertices in $T_i$ equals $|T_{i-1}|$, i.e, equals $n_{i-1}$. Together with Theorem~\ref{thm:LEXsort}, this means that line~10 during iteration $i+1$ executes in $O(n_i+n_{i-1})$ time. Similarly, line~11 executes in  $O(n'_i+n'_{i-1})$ time. In total, each iteration of the \textbf{for}-loop takes $O(n_i+n'_i+n_{i+1}+n'_{i+1})$ time, and summing over the number of iterations (i.e. the number of levels), we obtain a total running time of
\[\sum_{i=0}^{h}O(n_i+n'_i+n_{i+1}+n'_{i+1})=O(n+n'),\]
which is linear.
\end{proof}

\section{Final remarks}

An important remark about our run time result is that we have made a hidden assumption, namely that we will not, so to speak, run out of integers while running \alg{CondenseLabels}. For example, on a 64-bit computer, we may only represent integers up to $2^{64}-1$. Technically speaking, we have thus assumed that all considered rooted trees have an order bounded by, for example, $2^{64}$, or equivalently, that $\log_2(|V(T)|)=O(1)$ for all considered trees $T$. On the bright side, trees with more than $2^{64}>10^{18}$ vertices are rare in practice. Moreover, \cite{Campbell1991} discusses how the AHU-algorithm can be adapted to (in theory) work with arbitrarily large rooted trees of order $n$, and get a runtime of $O(n\log_2(n))$.

Secondly, it is worth mentioning that the AHU-algorithm can be used to detect isomorphisms of unrooted trees as well. The only additional work is that we need to decide on a canonical root for both input trees, and then run the AHU-algorithm on the rooted versions of the input. The best choice of root seems to be the so-called center (consisting of the vertex or the two vertices that lies on every path of maximal length). Not much has been written about this, but see the lecture notes \cite{Suderman2002} or the implementation in the Python package NetworkX \cite{Hagberg+2008} for additional hints.

\printbibliography

\end{document}